\documentclass[aps,pra,showpacs]{revtex4}
\usepackage{epsfig}
\usepackage{amsbsy,latexsym}
\usepackage{amsmath}
\usepackage{amssymb, mathrsfs}
\usepackage[mathscr]{eucal}
\usepackage{hyperref}
\usepackage{amsfonts}
\usepackage{amsthm}
\usepackage{amsmath}
\usepackage{amsfonts}
\usepackage{latexsym}
\usepackage{amssymb}
\usepackage{amscd}
\usepackage[latin1]{inputenc}
\usepackage{verbatim}

\usepackage[english]{babel}

\newtheorem{definition}{Definition}[section]
\newtheorem{theorem}[definition]{Theorem}
\newtheorem{lemma}[definition]{Lemma}
\newtheorem{corollary}[definition]{Corollary}
\newtheorem{proposition}[definition]{Proposition}
\theoremstyle{definition}

\newcommand\style{\mathsf }
\newcommand{\B}{\style{B}}
\newcommand{\M}{\style{M}}

\newcommand\A{{\style A}}
\renewcommand{\H}{\style{H}}

\newcommand{\T}{\style T}

\newcommand{\N}{\style{N}}

\newcommand{\lgG}{\rm{G}}        %%%% Lie group G
\newcommand{\lgGL}{{\rm{GL}}}        %%%% Lie group GL
\newcommand\osr{{\style R}}

\newcommand\oss{{\style S}}
\newcommand\ost{{\style T}}

\newcommand\tr{ \operatorname{Tr} } 

\begin{document}
\title{The Fidelity of Density Operators 
in an Operator-Algebraic Framework}

\author{Douglas Farenick}
%\affiliation{Department of Mathematics \& Statistics, University of Regina, Regina, Saskatchewan S4S 0A2, Canada}
 
\author{Samuel Jaques}

\author{Mizanur Rahaman}
\affiliation{Department of Mathematics and Statistics, University of Regina,
Regina, Saskatchewan S4S 0A2, Canada}

\date{ \today}

\begin{abstract}
Josza's definition of fidelity \cite{josza1994} for a pair of (mixed) quantum states  
is studied in the context of two types of operator algebras. The first setting is mainly algebraic in that it
involves unital C$^*$-algebras $\A$ that possess a faithful trace functional $\tau$. 
In this context, the role of quantum states (that is, density operators)
in the classical quantum-mechanical
framework is assumed by positive elements $\rho\in\A$ for which $\tau(\rho)=1$.
The second of our two settings is more operator theoretic: 
by fixing a faithful normal semifinite trace $\tau$ on a
semifinite von Neumann algebra $\M$, we define 
and consider the fidelity of pairs of
positive operators in $\M$ of unit trace.  
The main results of this paper address monotonicity and preservation of fidelity under the action of
certain trace-preserving positive linear maps of $\A$ or of the predual $\M_*$. Our results in the von Neumann algebra
setting are novel in that we focus on the Schr\"odinger picture rather than the Heisenberg picture, and they
also yield a new proof of a theorem of Moln\'ar \cite{molnar2001} on the structure of fidelity-preserving quantum
channels on the trace-class operators. 
 
\end{abstract}

\maketitle

 %%%%%%%%%%%%%%%%%%%%%%%%%%%%%%
\section*{Introduction}

In communication and information theory, the notion of 
fidelity provides a quantitative measure for the qualitative assessment of how well data or
information has been preserved through some type of transmission procedure or information processing task. 
Not surprisingly, this concept appears in quantum information theory 
\cite{berta--furrer--scholz2016,josza1994,uhlmann1976} as well, with the aim of providing a similar quantitative measure. 
In a rather different form (namely, in the guise of the Bures 
metric \cite{alberti--uhlmann1983,bures1969,kosaki1983}), the 
notion of fidelity also occurs in operator algebra theory.

In one of the most important settings, fidelity 
is a numerical measure of how close one state $\sigma$ of a quantum system is to another state $\rho$. 
For pure states---that is, for unit vectors $\xi$ and $\eta$ in a separable Hilbert space $\H$---the
fidelity is measured by $|\langle\xi,\eta\rangle|$, the modulus of the inner product of $\xi$ and $\eta$. 
Therefore, high fidelity occurs when the 
numerical measure is close to $1$, and at exactly $1$ the
vectors $\xi$ and $\eta$ are identical up to a phase factor. For mixed states---that is, for 
density operators $\sigma$ and $\rho$
acting on $\H$---there is a very satisfactory and useful notion of fidelity suggested by Josza \cite{josza1994}: 
namely, the quantity $F(\sigma,\rho)$ defined by
$\tr(|\sigma^{1/2}\rho^{1/2}|)$, where $\tr$ denotes the canonical trace (or tracial weight) on the algebra 
$\B(\H)$ of a bounded linear operators
acting on $\H$. In the case where $\sigma$ and $\rho$ are rank-1 density operators,
then the fidelity measure $\tr(|\sigma^{1/2}\rho^{1/2}|)$ coincides with $|\langle\xi,\eta\rangle|$,
where $\xi,\eta\in\H$ are unit vectors that span the ranges of $\sigma$ and $\rho$ respectively.

Recall that a \emph{channel}, or \emph{quantum channel}, is a bounded 
linear map ${\mathcal E}:\T(\H)\rightarrow\T(\H)$, where $\T(\H)$ is the
Banach space of trace-class operators on $\H$, such that
the dual map ${\mathcal E}^*:\B(\H)\rightarrow\B(\H)$ is normal, unital, and completely positive.
Here, the duality satisfies
\[
\tr\left({\mathcal E}(s)x\right)=\tr\left(s{\mathcal E}^*(x)\right),\;\mbox{ for all }s\in\T(\H),\;x\in\B(\H).
\]
By applying the Stinespring Theorem to ${\mathcal E}^*$ and by 
making use of the fact that ${\mathcal E}^*$ is unital and normal, 
Kraus \cite{kraus1971,Kraus-book}
showed that every channel ${\mathcal E}$ of $\T(\H)$ has the form
\[
{\mathcal E}(s)=\sum_{n}v_nsv_n^*,\;s\in \T(\H),
\]
for some countable set 
$\{v_n\}_n\subset\B(\H)$ such that $\sum_{n }v_n^*v_n=1$.

There is a well-known monotonicity property of channels:
if ${\mathcal E}:\T(\H)\rightarrow\T(\H)$ is a channel, then
\[
F(\sigma,\rho) \leq F\left({\mathcal E}(\sigma),{\mathcal E}(\rho)\right)
\]
for all states $\sigma$ and $\rho$. Now if ${\mathcal E}$ is a unitary channel, which is to say that there exists a unitary operator
$u$ on $\H$ such that ${\mathcal E}(s)=usu^*$ for every $s\in \T(\H)$, then in fact 
the channel ${\mathcal E}$ preserves fidelity in the sense that
$F(\sigma,\rho) = F\left({\mathcal E}(\sigma),{\mathcal E}(\rho)\right)$ for all density operators $\sigma$ and $\rho$. It is therefore quite important
and natural to determine whether unitary channels exhaust all cases of channels in which fidelity is preserved. This issue was resolved
by Moln\'ar \cite{molnar2001}\label{molnar}:
if ${\mathcal E}:\T(\H)\rightarrow\T(\H)$ is a surjective channel that preserves fidelity, then ${\mathcal E}$ is a unitary channel.

Our aim in this paper is to study channels and fidelity in the context of 
(i) unital C$^*$-algebras $\A$ that possess a faithful positive linear 
tracial functional $\tau:\A\rightarrow\mathbb C$ 
and (ii) semifinite von Neumann algebras $\M$. 
In the first of these two settings, the role of quantum states in the classical quantum-mechanical
framework is assumed by positive elements $\rho\in\A$ for which $\tau(\rho)=1$. 
We shall also observe that in some situations it is
possible to relax the (normally assumed) requirement for completely positivity of 
a channel to the much weaker condition that the 
map satisfy the Schwarz inequality.  A related study was undertaken by Timoney \cite{timoney2007}, where he
uses the term ``tracial geometric mean'' for what we are calling fidelity; however, Timoney's work and the results of this paper are quite different.

In the second scenario,
when considering quantum operations defined on von Neumann algebras, the first challenge is to define rigorously what is meant
by a quantum channel. The recent works of Crann and Neufang \cite{crann--neufang2013}, and 
of Crann, Kribs, and Todorov \cite{crann--kribs--todorov2016}, 
are examples of how to approach
the issue via the Heisenberg picture. Herein, we shall adopt an approach using the Schr\"odinger picture, 
making use of the identification of the fact that the predual $\M_*$ of an arbitrary von Neumann algebra $\M$ is a matrix-ordered
space, in the sense of Choi and Effros \cite{choi--effros1977}. 
A notion of fidelity for normal states of any von Neumann algebra $\M$ was put forward by Uhlmann
\cite{uhlmann1976} based on earlier ideas of Bures \cite{bures1969}. 
However, in this paper we are concerned with a different notion of fidelity in the von Neumann algebra
framework: namely, 
the one put forward by Jozsa \cite{josza1994}, which is defined through the use of a tracial weight.
Therefore, our concern in this paper will be with semifinite von Neumann algebras. For von Neumann algebras
acting on finite-dimensional Hilbert spaces, Uhlmann fidelity and Josza fidelity are the same quantities.
An analysis of 
various notions of fidelity in the general von Neumann algebra setting was undertaken by Alberti \cite{alberti2003}.

This first part of this paper is devoted to the study of fidelity in unital C$^*$-algebras that possess a faithful trace functional. 
The second half of the paper deals with the von Neumann algebra case. In the Schr\"odinger picture, it is necessary to consider
duality: specifically, for any von Neumann algebra $\M$, the dual of the matrix-ordered space $\M_*$ is the operator system $\M$
\cite{choi--effros1977}. In many applications, such as those where $\M=\B(\H)$, all normal completely positive linear maps of $\M$ are inner,
which is to say that each normal completely positive map admits a Kraus decomposition. In such cases, especially when $\H$ has
finite dimension, it is common in the literature to blur the distinction between the Heisenberg and Schr\"odinger pictures when 
working with quantum operations and channels; however, strictly speaking, the matrix ordered spaces $\T(\H)$ and $\B(\H)$ are not
the same, even though in finite dimensions they result in the same sets of operators. This distinction has consequences 
for the analysis of fidelity, and is discussed in greater detail in the second part of this paper.

%%%%%%%%%%%%%%%%%%%%%%%%%%%%%%
\section{Fidelity in a C$^*$-algebra Framework}

In this section it is assumed that $\A$ is a unital C$^*$-algebra. The cone of positive elements of $\A$ is
denoted by $\A_+$, while the real vector space of selfadjoint (or hermitian) elements of $\A$ is given by $\A_{\rm sa}$.
The notation $h\leq k$, for $h,k\in \A_{\rm sa}$, indicates that $k-h\in\A_+$.
Two elements $x,y\in\A$ are said to be \emph{orthogonal}
if $xy=yx=x^*y=xy^*=0$. The notation $x\bot y$ is to signify that $x$ and $y$ are orthogonal. 
The group of invertible elements of $\A$ is denoted by $\lgGL(\A)$ and the set $\lgGL(\A)_+$ of
positive invertible elements of $\A$ is defined by $\lgGL(\A)_+=\lgGL(\A)\cap\A_+$.

We shall also assume that $\A$ admits at least one faithful trace; that is, there exists
a continuous linear map $\tau:\A\rightarrow\mathbb C$  
such that, for all $x,y\in\A$,  (i) $\tau(xy)=\tau(yx)$,
(ii) $\tau(x^*x)\geq 0$, and (iii) $\tau(x^*x)=0$ only if $x=0$. 

There are a variety of multiplicative and 
order-preserving linear maps on $\A$ that shall be considered here.

\begin{definition} A linear map ${\mathcal E}:\A\rightarrow\A$ is:
\begin{enumerate}
\item a \emph{Jordan homomorphism}, if ${\mathcal E}(x^*)={\mathcal E}(x)^*$ for all $x\in \A$ and ${\mathcal E}(h^2)={\mathcal E}(h)^2$
for every  $h\in\A_{\rm sa}$;
\item  a \emph{homomorphism}, if ${\mathcal E}(x^*)={\mathcal E}(x)^*$ and ${\mathcal E}(xy)={\mathcal E}(x){\mathcal E}(y)$
for all $x,y\in\A$;
\item an \emph{automorphism}, if ${\mathcal E}$ is a bijective homomorphism such that ${\mathcal E}(1)=1$;
\item \emph{positive}, if ${\mathcal E}(\A_+)\subseteq\A_+$;
\item \emph{$n$-positive}, if ${\mathcal E}\otimes{\rm id}_{\M_n(\mathbb C)}:\A\otimes\M_n(\mathbb C)\rightarrow\A\otimes\M_n(\mathbb C)$ is positive;
\item \emph{completely positive}, if ${\mathcal E}$ is $n$-positive for all $n\in\mathbb N$;
\item a \emph{Schwarz map} if ${\mathcal E}(x^*){\mathcal E}(x)\leq{\mathcal E}(x^*x)$
for every $x\in\A$;
\item an \emph{order isomorphism}, if ${\mathcal E}$ is a bijection
and if both ${\mathcal E}$ and ${\mathcal E}^{-1}$ are positive.
\end{enumerate}
\end{definition}

A Schwarz map is necessarily positive, and so ${\mathcal E}(x^*)={\mathcal E}(x)^*$ for every $x\in\A$  \cite[p.~2]{Stormer-book}. 
Furthermore, for each $x\in\A$ the inequality
${\mathcal E}(x^*){\mathcal E}(x)\leq{\mathcal E}(x^*x)$ implies that 
$\|{\mathcal E}(x)\|^2\leq\|{\mathcal E}(x^*x)\|\leq\|{\mathcal E}\|\,\|x^*x\|=\|{\mathcal E}\|\,\|x\|^2$, and so $\|{\mathcal E}\|^2\leq\|{\mathcal E}\|$;
that is, every Schwarz map ${\mathcal E}$ satisfies $\|{\mathcal E}\|\leq 1$.
Conversely, every contractive 
$2$-positive linear map ${\mathcal E}:\A\rightarrow\A$ is a Schwarz map \cite[Corollary 1.3.2]{Stormer-book}. There are,
however, Schwarz maps that are not $2$-positive, such as the map ${\mathcal E}$ on $\M_2(\mathbb C)$ obtained by averaging
the transpose and the normalised trace \cite[Appendix A]{choi1980b}.

Our first result provides alternative ways of computing the tracial geometric mean,
$\tau(|\sigma^{1/2}\rho^{1/2}|)$, as defined in \cite{timoney2007},
of positive elements $\sigma,\rho\in\A$. The characterisation (\ref{eq:var1}) below
is inspired by an approach of Watrous \cite{Watrous-book}.

\begin{proposition}\label{trace properties} If $a,\sigma,\rho\in\A_+$, then
\begin{equation}\label{eq:trace pos}
\tau(a) = \frac{1}{2}\inf_{y\in\lgGL(\A)_+}\left(\tau(ay)+\tau(ay^{-1})\right),
\end{equation}
\begin{equation}\label{eq:var1}
\tau(|\sigma^{1/2}\rho^{1/2}|)= \frac{1}{2}\inf_{y\in\lgGL(\A)_+}\left(\tau(\rho y)+\tau(\sigma y^{-1})\right)=\tau(|\rho^{1/2}\sigma^{1/2}|),
\end{equation}
and
\begin{equation}\label{eq:var2}
\tau(|\sigma^{1/2}\rho^{1/2}|)= \displaystyle\frac{1}{4}\inf_{y\in\lgGL(\A)_+} \left[
\tau(\sigma y)+\tau(\sigma y^{-1})+ \tau(\rho y)+\tau(\rho y^{-1})  \right] .
\end{equation}
\end{proposition}

\begin{proof}
If $a\in\A_+$ and $y\in\lgGL(\A)_+$, then
both $\tau(ay)=\tau(a^{1/2}ya^{1/2})$ and $\tau(ay^{-1})=\tau(a^{1/2}y^{-1}a^{1/2})$
are nonnegative real numbers. Because 
\[
y+y^{-1}-2=(y^{\frac{1}{2}}-y^{\frac{-1}{2}})^2\in \A_+, 
\]
we have that
\[
a^{1/2}ya^{1/2}+a^{1/2}y^{-1}a^{1/2}\geq 2 a.
\]
Thus, 
\[
\tau(a)\leq \frac{1}{2}\inf_{y\in\lgGL(\A)_+}\left(\tau(ay)+\tau(ay^{-1})\right).
\]
In taking $y=1$, the infimum above is attained and yields
\[
\tau(a) = \frac{1}{2}\inf_{y\in\lgGL(\A)_+}\left(\tau(ay)+\tau(ay^{-1})\right),
\]
which proves equation (\ref{eq:trace pos}).

Suppose now that each of the $\tau$-states $\sigma$ and $\rho$ is invertible. Thus, 
$b=|\rho^{1/2}\sigma^{1/2}|\in\lgGL(\A)_+$
and the (completely) positive linear map $\Psi$ on $\A$ defined by
$\Psi(x)=b^{-1/2}\rho^{1/2}x\rho^{1/2}b^{-1/2}$, for $x\in \A$, is a bijection of $\lgGL(\A)_+$ with itself.
Furthermore, if $y\in\lgGL(\A)_+$, then
\[
\tau\left(\Psi(y)b\right)=\tau(\rho y)\,\mbox{ and }\, \tau\left(\Psi(y)^{-1}b\right)=\tau(\sigma y^{-1}).
\]
Hence, by equation (\ref{eq:trace pos}),
\[
\begin{array}{rcl}
 \tau(b)&=&\displaystyle \frac{1}{2}\inf_{y\in\lgGL(\A)_+}\left(\tau(by)+\tau(by^{-1})\right) \\ &&\\
&=&\displaystyle \frac{1}{2}\inf_{y\in\lgGL(\A)_+}\left(\tau\left(\Psi(y)b\right)+\tau\left(\Psi(y)^{-1}b\right)\right) \\&& \\
&=&\displaystyle \frac{1}{2}\inf_{y\in\lgGL(\A)_+}\left(\tau(\rho y)+\tau(\sigma y^{-1})\right).
\end{array}
\]
If it so happens that one of $\sigma$ or $\rho$ is not invertible, then they may be replaced by the 
positive invertible elements $\sigma_\varepsilon=\sigma+\varepsilon 1$
and $\rho_\varepsilon=\rho+\varepsilon 1$ to obtain for 
$b_\varepsilon=|\rho_\varepsilon^{1/2}\sigma_\varepsilon^{1/2}|\in\lgGL(\A)_+$ that
\[
 \tau(b_\varepsilon) =  \frac{1}{2}\inf_{y\in\lgGL(\A)_+}\left(\tau(\rho_\varepsilon y)+\tau(\sigma_\varepsilon y^{-1})\right).
 \]
Because $\tau(b)<\tau(b_\varepsilon)$ and because $\tau(\rho_\varepsilon y)+\tau(\sigma_\varepsilon y^{-1})$ decreases to
$(\tau(\rho y)+\tau(\sigma y^{-1})$ as $\varepsilon$ decreases to $0$, we have that
\[
\begin{array}{rcl}
\tau(b) &\leq& \displaystyle\frac{1}{2}\inf_{\varepsilon>0}\left[\inf_{y\in\lgGL(\A)_+}\left(\tau(b_\varepsilon y)+\tau(b_\varepsilon y^{-1})\right)\right]
\\&&\\
&=& \displaystyle\frac{1}{2}\inf_{\varepsilon>0}\left[\inf_{y\in\lgGL(\A)_+}\left(\tau(\rho_\varepsilon y)+\tau(\sigma_\varepsilon y^{-1})\right)\right].
\end{array}
\]
Again, using $y=1$, we obtain
\[
F_\tau(\sigma,\rho)=\tau(b)= \frac{1}{2}\inf_{y\in\lgGL(\A)_+}\left(\tau(\rho y)+\tau(\sigma y^{-1})\right),
\]
which establishes the first part of equation (\ref{eq:var1}). The second part of equation (\ref{eq:var1})
follows from the first by using the fact that the map
$y\mapsto y^{-1}$ is a bijection $\lgG(\A)_+\rightarrow\lgG(\A)_+$.

To prove that $\tau(|\sigma^{1/2}\rho^{1/2}|)= \displaystyle\frac{1}{4}\inf_{y\in\lgGL(\A)_+} \left[
\tau(\sigma y)+\tau(\sigma y^{-1})+ \tau(\rho y)+\tau(\rho y^{-1})  \right] $, suppose once again that $\sigma$ and $\rho$ are
invertible. With $b=|\sigma^{1/2}\rho^{1/2}|$ and $c=|\rho^{1/2}\sigma^{1/2}|$, the positive linear maps $\Psi,\Phi:\A\rightarrow\A$ 
defined by
\[
\Psi(x)=b^{-1/2}\rho^{1/2}x\rho^{1/2}b^{-1/2}\;\mbox{ and }\;
\Phi(x)=c^{-1/2}\sigma^{1/2}x\sigma^{1/2}c^{-1/2},
\]
for $x\in\A$, are bijections. Moreover,
\[
\begin{array}{rcllcl}
\tau\left(\Psi(y)b\right) &=&    \tau\left(\rho y\right),  &\qquad \tau\left(\Psi(y)^{-1}b\right) &=&    \tau\left(\sigma y^{-1}\right), \\ &&&&& \\
\tau\left(\Phi(y)c\right) &=&    \tau\left(\sigma y\right),  &\qquad \tau\left(\Phi(y)^{-1}c\right) &=&    \tau\left(\rho y^{-1}\right) .
\end{array}
\]
Thus, the relations above and the fact that $\Psi$ is a positive linear bijection together imply that
\[
\begin{array}{rcl}
2\tau(b) &=& \displaystyle \inf_{y\in\lgGL(\A)_+}\left(\tau(b y)+\tau(b y^{-1})\right) \\ && \\
&=&  \displaystyle \inf_{y\in\lgGL(\A)_+}\left(\tau(b \Psi(y))+\tau(b \Psi(y)^{-1})\right) \\ && \\
&=& \displaystyle \inf_{y\in\lgGL(\A)_+}\left(\tau\left(\rho y\right)+\tau\left(\sigma y^{-1}\right)\right).
\end{array}
\]
Similarly,
\[
2\tau(c) =  \displaystyle \inf_{y\in\lgGL(\A)_+}\left(\tau\left(\sigma y\right)+\tau\left(\rho y^{-1}\right)\right).
\]
Hence,
\[
\begin{array}{rcl}
2\tau(b)+2\tau(c) &=& \displaystyle \inf_{y\in\lgGL(\A)_+}\left[\tau\left(\rho y\right)+\tau\left(\sigma y^{-1}\right)\right]
                           +  \displaystyle \inf_{y\in\lgGL(\A)_+}\left[\tau\left(\sigma y\right)+\tau\left(\rho y^{-1}\right)\right]\\ && \\
&\leq&  \displaystyle \inf_{y\in\lgGL(\A)_+}\left[\tau\left(\rho y\right)+\tau\left(\sigma y^{-1}\right)\right)
                           +   \left(\tau\left(\sigma y\right)+\tau\left(\rho y^{-1}\right)\right] \\ && \\
&=& \displaystyle \inf_{y\in\lgGL(\A)_+}\left(\tau\left((b+c)y\right)+\tau\left((b+c)y^{-1}\right)\right) \\&& \\
&=& 2\tau(b+c).
\end{array}
\]
Therefore, the intermediate inequality above is an equality. Hence,
\[
 \displaystyle \inf_{y\in\lgGL(\A)_+}\left[\tau\left(\rho y\right)+\tau\left(\sigma y^{-1}\right)\right)
                           +   \left(\tau\left(\sigma y\right)+\tau\left(\rho y^{-1}\right)\right]
=2\left( \tau(|\sigma^{1/2}\rho^{1/2}|)+\tau(|\rho^{1/2}\sigma^{1/2}|)\right)=4\tau(|\sigma^{1/2}\rho^{1/2}|),
\]
which proves equation (\ref{eq:var2}) in the case where $\sigma$ and $\rho$ are invertible.

If one of $\sigma$ or $\rho$ is not invertible, then let $\sigma_\varepsilon=\sigma+\varepsilon 1$
and $\rho_\varepsilon=\rho+\varepsilon 1$, for $\varepsilon>0$, so that
%$b_\varepsilon=|\sigma_\varepsilon^{1/2}\rho_\varepsilon^{1/2}|$ and $c_\varepsilon=|\rho_\varepsilon^{1/2}\rho_\varepsilon^{1/2}|$; therefore,
\[
4\tau(|\sigma_\varepsilon^{1/2}\rho_\varepsilon^{1/2}|)=
\displaystyle \inf_{y\in\lgGL(\A)_+}\left[\tau\left(\rho_\varepsilon y\right)+\tau\left(\sigma_\varepsilon y^{-1}\right)\right)
                           +   \left(\tau\left(\sigma_\varepsilon y\right)+\tau\left(\rho_\varepsilon y^{-1}\right)\right]
\]
Since the traces of the ``$\varepsilon$-elements'' decrease as $\varepsilon\rightarrow0^+$, 
the equation above also holds for $\sigma$ and $\rho$, which completes the proof of 
equation (\ref{eq:var2}).
\end{proof}

Motivated by the use of the term ``state'' for density operators, 
we shall call the elements of the set $\oss_\tau$, defined by
\[
\oss_\tau=\{\rho\in\A_+\,|\,\tau(\rho)=1\},
\]
\emph{$\tau$-states} of $\A$.
Thus, in this terminology,
a  $\tau$-state is a positive element of $\A$ with trace $1$, and is not to be confused with the traditional meaning of the word ``state'' in
C$^*$-algebra theory, which refers to a positive linear functional of norm $1$.

\begin{definition} The \emph{$\tau$-fidelity} of a pair of $\tau$-states $\sigma,\rho\in\A_+$ is the
quantity denoted by $F_\tau(\sigma,\rho)$ and defined by
\[
F_\tau(\sigma,\rho)=\tau(|\sigma^{1/2}\rho^{1/2}|).
\]
\end{definition} 

The essential properties of $\tau$-fidelity are described below. 

\begin{theorem}\label{miza} 
If $\sigma,\rho\in \A$ are a pair of $\tau$-states, then
\begin{enumerate}
\item $F_\tau(\sigma,\rho)=F_\tau(\rho,\sigma)$, 
\item $0\leq F_\tau(\sigma,\rho) \leq 1$,
\item $F_\tau(\sigma,\rho)=0$ if and only if $\sigma\bot\rho$, and
\item $F_\tau(\sigma,\rho)=1$ if and only if $\sigma=\rho$.
\end{enumerate} 
Furthermore, if ${\mathcal E}:\A\rightarrow\A$ is a positive linear map such that $\tau\circ{\mathcal E}=\tau$, then
\begin{equation}\label{e:fidelity incr}
F_\tau(\sigma,\rho)\leq F_\tau\left( {\mathcal E}(\sigma),{\mathcal E}(\rho)\right).
\end{equation}
\end{theorem}

\begin{proof}   Assertion (1) is an immediate consequence of 
equation (\ref{eq:var1}) of Proposition \ref{trace properties}.

To prove (2), note that $0\leq F_\tau(\sigma,\rho)$ is obvious. On the other hand, 
the tracial arithmetic-geometric mean inequality \cite[Theorem 2.4]{farenick--manjegani2005} yields
\[
 F_\tau(\sigma,\rho) = \tau\left(|\sigma^{1/2}\rho^{1/2}|\right)\leq \frac{\tau(\sigma)+\tau(\rho)}{2} =1.
 \]
Thus, $0\leq F_\tau(\sigma,\rho) \leq 1$.

Statement  (3) and the sufficiency of statement (4) are straightforward, and so we prove only the necessity of statement (4). 
If $F_\tau(\sigma,\rho)=1$, then $\tau\left(|\sigma^{1/2}\rho^{1/2}|\right)= \frac{1}{2}\left(\tau(\sigma)+\tau(\rho)\right)$
is a case of equality in the tracial arithmetic-geometric mean inequality in unital C$^*$-algebras. By \cite[Theorem 3.6]{farenick--manjegani2005}, equality
is achieved only if $\sigma=\rho$.  

To prove inequality (\ref{e:fidelity incr}), 
suppose that ${\mathcal E}:\A\rightarrow\A$ is a positive linear map such that $\tau\circ{\mathcal E}=\tau$.
With $y=1$, equation (\ref{eq:var1}) yields 
\[
F_\tau(\sigma,\rho) \leq \frac{1}{2}\left(\tau(\rho)+\tau(\sigma)\right)= \frac{1}{2}\left(\tau({\mathcal E}(\rho))+\tau({\mathcal E}(\sigma))\right).
\]
Equation (\ref{eq:var1})
yields
\[
\begin{array}{rcl}
2F_\tau(\sigma,\rho)  &\leq&
 \displaystyle\frac{1}{2}\inf_{y\in\lgGL(\A)_+} \left[
\tau({\mathcal E}(\rho) y)+\tau({\mathcal E}(\rho) y^{-1})+ \tau({\mathcal E}(\sigma) y)+\tau({\mathcal E}(\sigma) y^{-1})  \right]  \\&&\\
&=& F_\tau({\mathcal E}(\rho),{\mathcal E}(\sigma))+F_\tau({\mathcal E}(\sigma),{\mathcal E}(\rho)) \\ && \\
&=& 2 F_\tau({\mathcal E}(\sigma),{\mathcal E}(\rho)).
\end{array}
\]
Hence, $F_\tau(\sigma,\rho) \leq F_\tau\left({\mathcal E}(\rho),{\mathcal E}(\sigma)\right)$, which establishes inequality (\ref{e:fidelity incr}).  
\end{proof}

\begin{corollary} The function $d_{B,\tau}$ on $\oss_\tau$ defined by
\[
d_{B,\tau}(\sigma,\rho)=\sqrt{1-F_\tau(\sigma,\rho)}
\]
is a metric.
\end{corollary}

The metric $d_{B,\tau}$ above is another instance of the
\emph{Bures metric}, which has been studied in a number of operator algebraic settings.

The proof of
Theorem \ref{miza} above is inspired by the approach taken by Watrous \cite{Watrous-book} in the case of matrix algebras.

\begin{definition} A positive linear map ${\mathcal E}:\A\rightarrow\A$ \emph{preserves $\tau$-fidelity} if
\begin{enumerate}
%\item ${\mathcal E}$ is positive,
\item $\tau\circ{\mathcal E}=\tau$, and 
\item $F_\tau\left({\mathcal E}(\sigma),{\mathcal E}(\rho)\right)=F_\tau(\sigma,\rho)$
for all $\tau$-states $\sigma,\rho\in\A$.
\end{enumerate}
\end{definition}

\begin{lemma}\label{order iso} If a positive linear map ${\mathcal E}:\A\rightarrow\A$ preserves $\tau$-fidelity, then ${\mathcal E}$
is an injection. Moreover, if ${\mathcal E}$ is also surjective, then ${\mathcal E}$ is an order isomorphism.
\end{lemma} 

\begin{proof} 
We first prove that ${\mathcal E}$ is a linear injection. To this end, 
if $a,b\in\A_+$ satisfy ${\mathcal E}(a)={\mathcal E}(b)$, then 
setting $\sigma=\tau(a)^{-1}a$ and $\rho=\tau(b)^{-1}b$ yields elements $\sigma,\rho\in\oss_\tau$ for which
$\tau(a)\sigma=a$ and $\tau(b)\rho=b$. Thus, ${\mathcal E}(a)={\mathcal E}(b)$
yields $\tau(a){\mathcal E}(\sigma)=\tau(b){\mathcal E}(\rho)$ and, by applying the trace to this last equation,
$\tau(a)=\tau(b)$. Hence, ${\mathcal E}(\sigma)={\mathcal E}(\rho)$ and, therefore,
\[
1=F_\tau\left({\mathcal E}(\sigma),{\mathcal E}(\rho)\right)=\tau(\sigma,\rho).
\]
Theorem \ref{miza} implies that $\sigma=\rho$ and, consequently, that
$a=b$. Therefore, ${\mathcal E}_{\vert\A_+}$ is an injective function.

Assume that $h\in\A$ is selfadjoint and that ${\mathcal E}(h)=0$. There exist $h_+,h_-\in \A_+$ such that 
$h=h_+-h_-$ and so $0={\mathcal E}(h_+)-{\mathcal E}(h_-)$. In other words, ${\mathcal E}(h_+)={\mathcal E}(h_-)$ and so $h_+=h_-$ because
${\mathcal E}_{\vert\A_+}$ is an injective function. Thus, $h=0$. Because ${\mathcal E}$ is real linear, this implies that ${\mathcal E}_{\vert\A_{\rm sa}}$
is an injective function. Lastly, suppose that $x\in \A$ satisfies ${\mathcal E}(x)=0$. 
Thus, ${\mathcal E}(x^*)=0$ and, writing $x=a+ib$ for some $a,b\in\A_{\rm sa}$, 
${\mathcal E}(a)+i{\mathcal E}(b)={\mathcal E}(a)-i{\mathcal E}(b)=0$. Thus, ${\mathcal E}(a)={\mathcal E}(b)=0$, which implies that $a=b=0$ and $x=0$.
Hence ${\mathcal E}$ is a linear injection.

Because ${\mathcal E}$ is a linear injection and assumed to be surjective, ${\mathcal E}$ admits a linear inverse ${\mathcal E}^{-1}$.
We aim to show that ${\mathcal E}^{-1}$ is a positive map. 
To this end, select $h\in\A_+$ and let $a={\mathcal E}^{-1}(h)$. Therefore, because ${\mathcal E}$ is an injection,
\[
{\mathcal E}(a)=h=h^*={\mathcal E}(a)^*={\mathcal E}(a^*)
\]
implies that $a=a^*$. Thus, there are $a_+,a_-\in\A_+$ such that $a=a_+-a_-$ and $a_+\bot a_-$.
Let $b={\mathcal E}(a_+)$ and $c={\mathcal E}(a_-)$ to obtain $b,c\in\A_+$. If one of $a_+$ or $a_-$ is zero, then 
$bc=cb=0$. If neither $a_+$ nor $a_-$ is zero, then 
scale them by their traces so that $\sigma=\tau(a_+)^{-1}a_+$ and $\rho=\tau(a_-)^{-1}a_-$
are $\tau$-states. Because $\sigma\bot\rho$  and  
$0=F_\tau(\sigma,\rho)=F_\tau\left({\mathcal E}(\sigma),{\mathcal E}(\rho)\right)$, 
Theorem \ref{miza} yields ${\mathcal E}(\sigma)\bot {\mathcal E}(\rho)$. Therefore, by linearity,
${\mathcal E}(a_+)\bot{\mathcal E}(a_-)$, which shows that $bc=cb=0$.

From $b=h+c$ we obtain
\[
0=bc=h c+c^2 \,\mbox{ and }\, 0=cb=ch+c^2.
\]
Hence, $h c-ch=0$. Because $h$ and $c$ commute, so do $h$ and $c^{1/2}$. Thus,
\[
0=bc= h c+c^2=c^{1/2}h c^{1/2} +c^2
\]
expresses $0$ as a sum of positive elements $c^{1/2}h c^{1/2}$ and $c^2$. Therefore
$c^2= h c=0$, and so $c=0$ also. Hence, $0={\mathcal E}^{-1}(c)=a_-=0$, which yields $a=a_+\in\A_+$.
Hence, ${\mathcal E}^{-1}$ is a positive map, which implies that
${\mathcal E}$ is an order isomorphism.
\end{proof}

\begin{corollary} Every surjective positive linear map that preserves $\tau$-fidelity is a Jordan
isomorphism.
\end{corollary}

\begin{proof} By Kadison's Theorem \cite[Theorem 2.1.3]{Stormer-book}, every order isomorphism of a unital 
C$^*$-algebra is a Jordan isomorphism.
\end{proof}

Jordan isomorphisms are quite general, and somewhat more information can be elicited in cases where the C$^*$-algebra 
$\A$ resembles a matrix algebra in terms of certain algebraic properties.

\begin{definition} A unital C$^*$-algebra $\A$ is:
\begin{enumerate}
\item \emph{finite}, if 
$xy=1$ implies $yx=1$, and
\item \emph{quasi-transitive}, if $x\A y=\{0\}$, for some $x,y\in\A$, holds only if $x=0$ or $y=0$.
\end{enumerate}
\end{definition}

Theorem \ref{main result 1} is the first main result of the present paper.

 \begin{theorem}\label{main result 1} Assume that $\A$ is a finite quasi-transitive C$^*$-algebra.
If a surjective Schwarz map ${\mathcal E}:\A\rightarrow\A$ preserves $\tau$-fidelity, 
then ${\mathcal E}$ is an automorphism of $\A$.
\end{theorem}

\begin{proof} 
Because $\|{\mathcal E}\|\leq 1$, the element ${\mathcal E}(1)\in\A$ satisfies $\|{\mathcal E}(1)\|\leq 1$.
Thus, ${\mathcal E}(1)$ is a positive contraction and, therefore, $1-{\mathcal E}(1)$ is positive. Apply $\tau$ to obtain
$0\leq\tau(1-{\mathcal E}(1))=\tau(1)-\tau\circ{\mathcal E}(1)=\tau(1)-\tau(1)=0$. Because $\tau$ is faithful and $1-{\mathcal E}(1)\in\A_+$, we
have $\tau(1-{\mathcal E}(1))=0$ only if $1-{\mathcal E}(1)=0$. That is, ${\mathcal E}$ is unital.

Since, by Lemma \ref{order iso}, ${\mathcal E}^{-1}$ and ${\mathcal E}$ are positive linear maps, their norms are achieved at $1\in\A$
\cite{russo--dye1966}, \cite[Theorem 1.3.3]{Stormer-book}. Hence,
${\mathcal E}^{-1}(1)=1$ implies that $\|{\mathcal E}^{-1}\|=1$. Thus, for every $x\in \A$,
\[
\|x\| = \|{\mathcal E}^{-1}\circ{\mathcal E}(x)\| \leq \|{\mathcal E}^{-1}\|\,\|{\mathcal E}(x)\|=\|{\mathcal E}(x)\|\leq\|x\|.
\]
In other words, ${\mathcal E}$ is an isometry.  

Select a unitary $u\in \A$. Because ${\mathcal E}$ is an isometry, we have that ${\mathcal E}(u)$ is in the closed unit ball of $\A$. Suppose that
${\mathcal E}(u)=\frac{1}{2}x+\frac{1}{2}y$ for some $x,y\in \A$ with $\|x\|\leq1$ and $\|y\|\leq1$. Then $u=\frac{1}{2}{\mathcal E}^{-1}(x)+\frac{1}{2}{\mathcal E}^{-1}(y)$.
Because ${\mathcal E}^{-1}$ is an isometry and every unitary in $\A$ is an extreme point of the closed unit ball of $\A$, $u={\mathcal E}^{-1}(x)={\mathcal E}^{-1}(y)$,
which implies that $x=y={\mathcal E}(u)$. Hence, ${\mathcal E}(u)$ is an extreme point of the closed unit ball of $\A$.
By \cite[Theorem 1]{kadison1951}, the extreme point
${\mathcal E}(u)$ is necessarily a partial isometry such that, if $e={\mathcal E}(u)^*{\mathcal E}(u)$ and $f={\mathcal E}(u){\mathcal E}(u^*)$, then $(1-e)\A(1-f)=\{0\}$.
By the hypothesis that $\A$ is quasi-transitive, we obtain $e=1$ or $f=1$. Thus, either ${\mathcal E}(u)^*{\mathcal E}(u)=1$
or ${\mathcal E}(u^*){\mathcal E}(u)=1$. Because $\A$ is finite, either of these conditions imply that ${\mathcal E}(u)$ is unitary.
Hence, ${\mathcal E}(u)$ is unitary for each unitary $u\in\A$.
Now for every unitary $u\in \A$, we see that 
\[
1={\mathcal E}(u^*u)={\mathcal E}(u^*){\mathcal E}(u)={\mathcal E}(u){\mathcal E}(u^*)={\mathcal E}(uu^*).
\] 
Hence, every
unitary in $\A$ is in the multiplicative domain $\mathfrak M_{\mathcal E}$ of ${\mathcal E}$. Because ${\mathcal E}$ is a Schwarz map, its multiplicative domain is a
C$^*$-subalgebra of $\A$ \cite[Corollary 2.1.6]{Stormer-book}. Thus, $\mathfrak M_{\mathcal E}$ contains the norm-closed span of the unitary group of $\A$, 
which by the Russo--Dye Theorem \cite[Theorem 1]{russo--dye1966} implies that $\mathfrak M_{\mathcal E}=\A$. Hence, the bijective unital map
${\mathcal E}$ is a homomorphism.
\end{proof}

To address the cases in which $\A$ may not be finite or quasi-transitive, 
we shall need to require more positivity of ${\mathcal E}$ than simply the Schwarz map feature.

\begin{definition} A linear map ${\mathcal E}:\A\rightarrow\A$ is said
to be a transformation of \emph{order zero}  if ${\mathcal E}(x)\bot{\mathcal E}(y)$ for all $x,y\in\A$ for which $x\bot y$.
\end{definition}

Although Gardner \cite{gardner1979a} has studied order zero maps in a fairly general context, the following results
draw upon the structure theory for completely positive order zero maps developed by 
Winter and Zacharias \cite{winter--zacharias2009}.

\begin{theorem}\label{main result 2} If ${\mathcal E}:\A\rightarrow\A$ that preserves $\tau$-fidelity, then 
there is a homomorphism $\pi:\A\rightarrow\A$ and a positive element $h$ in the centre of $\A$ such that
$\mathcal E(x)=\pi(x)h$, for every $x\in\A$. In particular, if
$\A$ has trivial centre, then
${\mathcal E}$ is an injective unital
homomorphism.
\end{theorem}

\begin{proof}  
Select any nonzero $a,b\in\A_+$ and scale them by their traces so that $\sigma=\tau(a)^{-1}a$ and $\rho=\tau(b)^{-1}b$
are $\tau$-states. Further, assume that $a\bot b$; thus, $\sigma\bot\rho$. Hence, by Theorem \ref{miza},
$0=F_\tau(\sigma,\rho)=F_\tau\left({\mathcal E}(\sigma),{\mathcal E}(\rho)\right)$ yields ${\mathcal E}(\sigma)\bot {\mathcal E}(\rho)$ and, by linearity,
${\mathcal E}(a)\bot{\mathcal E}(b)$. By Stinespring's Theorem, 
${\mathcal E}(x^*){\mathcal E}(x)\leq\|{\mathcal E}(1)\|^2{\mathcal E}(x^*x)$
for every $x\in\A$. Therefore, using the proof of Remark 2.4 in \cite{winter--zacharias2009}, we see that the property
``${\mathcal E}(a)\bot {\mathcal E}(b)$ for all positive $a$ and $b$ with $a\bot b$'' implies the property
``${\mathcal E}(x)\bot {\mathcal E}(y)$
for all $x,y\in \A$ for which $x\bot y$.'' That is, ${\mathcal E}$ is a completely positive linear map of order zero.

Because ${\mathcal E}$ is a completely positive linear map of order zero, $h={\mathcal E}(1)$ is the centre of $\A$
and there is a homomorphism $\pi:\A\rightarrow\A$ such that ${\mathcal E}(x)=\pi(x)h$, for all $x\in\A$
\cite[Theorem 3.3]{winter--zacharias2009}. If the centre of $\A$ is trivial, then $h=\lambda 1$
for some $\lambda\in\mathbb R_+$. Further, $\tau(1)=\tau\circ{\mathcal E}(1)=\tau\left(\pi(1)h\right)=\lambda\tau(1)$
implies that $\lambda=1$, and so ${\mathcal E}=\pi$. The fact that $\pi$ is an injection follows from the fact that
${\mathcal E}$ is an injection (Lemma \ref{order iso}). 
\end{proof}

The results of this section apply to finite von Neumann algebras $\N$, 
since such algebras are unital C$^*$-algebras that possess
a faithful trace functional. Indeed, if $\N$ is a factor, then $\N$ 
has a unique (normal) faithful trace functional $\tau$ satisfying $\tau(1)=1$,
and all other tracial functionals on $\N$ are positive scalar multiples of $\tau$.
Therefore, if a positive
linear map on $\N$ preserves fidelity with respect to one trace, then it does so with respect to every trace.  

\begin{proposition}\label{main result 1 finite factor} Assume that $\N$ is a finite factor. 
If ${\mathcal E}:\N\rightarrow\N $ is a surjective Schwarz map that
preserves fidelity, then ${\mathcal E}$ is an automorphism of $\N$. 
If, moreover, $\N$ is a finite-dimensional factor, then
${\mathcal E}$ is a unitary channel.
\end{proposition} 

\begin{proof} Because $\N$ is both finite and quasi-transitive as a C$^*$-algebra,
Theorem \ref{main result 1} implies that ${\mathcal E}$ is a necessarily an automorphism of $\N$ (and, moreover, automatically normal). 
If $N\cong\M_d(\mathbb C)$ for some $d\in\mathbb N$, then every automorphism of $\N$ is
inner; hence, ${\mathcal E}$ is a unitary channel.
\end{proof}

%%%%%%%%%%%%%%%%%%%%%%
\section{Fidelity in a von Neumann Algebra Framework}

Even though Proposition \ref{main result 1 finite factor} applies to a finite von Neumann algebra $\N$, 
the result makes no reference to the predual $\N_*$,  and therefore 
Proposition \ref{main result 1 finite factor} should be viewed as a result in the Heisenberg picture.
The goal of the present section is to develop a framework for the Schr\"odinger picture, and to then specialise
to the setting of semifinite von Neumann algebras for the purpose of analysing fidelity.

A useful fact about $2\times 2$ matrices over von Neumann algebras is recorded below for later reference.

\begin{lemma}\label{matrix lemma} {\rm (\cite[p.~166]{Takesaki-bookI})}
If $\M$ is a von Neumann algebra acting on $\H$ and if $a,b\in\M_+$, 
then the following statements are equivalent for $x\in\M$:
\begin{enumerate}
\item $\left[\begin{array}{cc} a&x\\ x^*&b\end{array}\right]$ is a positive operator on $\H\oplus\H$;
\item $x=a^{1/2}y b^{1/2}$ for some $y\in\M$ with $\|y\|\leq 1$.
\end{enumerate}
\end{lemma}

%%%%%%%%%%%%%%%%%
\subsection{Duality and Channels}
We begin by making sense of the commonly used phrase
``$\phi$ is a trace-preserving completely positive linear map on $\T(\H)$.'' On the one hand,
the phrase suggests that $\T(\H)$ is a treated as a matrix-order space; on the other hand, 
much of the literature interprets the phrase to mean that $\phi$ is a normal completely 
positive linear map on $\B(\H)$ that preserves the trace of operators in $\T(\H)$. 
These two interpretations are not entirely compatible.

We begin by reviewing the matrix order 
on the predual $\M_*$ of an arbitrary von Neumann algebra $\M$. Recall from \cite{choi--effros1977} that
if $\oss$ is a matrix order space, and if the positive cone of $\M_n(\oss)$ is denoted by $\M_n(\oss)_+$,
then a linear map $\phi:\oss\rightarrow\ost$ of matrix ordered spaces is $k$-positive if $\phi\left(\M_n(\oss)_+\right)\subseteq
\M_n(\ost)_+ $ for every $n=1,\dots, k$, and $\phi$ is completely positive if $\phi$ is $k$-positive for all $k\in\mathbb N$.

The matrix order on $\M$ is inherited from $\B(\H)$: 
a matrix $X=[x_{ij}]_{i,j}\in\M_n(\M)$ is positive if $X$ is a positive operator
on the $n$-fold direct sum $\displaystyle\bigoplus_1^n\H$.
A matrix $\Omega=[\omega_{ij}]_{i,j}\in \M_n(\M_*)$ is positive if
\begin{equation}\label{eq:matrix dual}
\sum_{i=1}^n\sum_{j=1}^n\omega_{ij}(x_{ij})\geq 0,
\end{equation}
for every $X=[x_{ij}]_{i,j}\in\M_n(\M)$. Let $\M_n(\M_*)_+$ denote the positive matrices over $\M_*$.
(Recall that $\M_*$ consists of normal linear functionals on $\M$.)
A useful criterion for membership in $\M_n(\M_*)_+$ is as follows:
$\Omega=[\omega_{ij}]_{i,j}\in \M_n(\M_*)_+$ if and only if the linear map $\mathfrak L_\Omega:\M\rightarrow\M_n$,
defined by
\begin{equation}\label{lemma 4.7}
\mathfrak L_\Omega(x)=\left[ \omega_{ij}(x)\right]_{i,j},
\end{equation}
for $x\in\M$, is a completely positive map \cite[Lemma 4.7]{choi--effros1977}.

In the case where $\M=\M_d(\mathbb C)$, then every $\omega\in \M_*$ is determined uniquely by some matrix 
$y_\omega\in\M$ via the formula $\omega(x)=\tr(xy_\omega)$, 
for $x\in\M $. It can happen that a matrix 
$\Omega=[\omega_{ij}]_{i,j=1}^n $ of normal linear functionals is positive, yet the corresponding matrix of operators 
$Y_\Omega=[y_{\omega_{ij}}]_{i,j=1}^n$ fails to be positive as an operator on $\displaystyle\bigoplus_1^n\mathbb C^d$.
For example, let $d=2$ and $\M=\M_2(\mathbb C)$, and consider the matrix  
$\Omega=\left[\begin{array}{cc} \omega_{11}& \omega_{12} \\ \omega_{21}& \omega_{22} \end{array}\right]\in \M_2(\M_*)$, where
the corresponding operators $y_{\omega_{ij}}\in\M$ inducing each $\omega_{ij}\in\M_*$ are
\[
y_{\omega_{11}}=\left[\begin{array}{cc} 1& 0 \\ 0& 0 \end{array}\right], \;
y_{\omega_{12}}=\left[\begin{array}{cc} 0& 0 \\ 1 & 0 \end{array}\right], \;
y_{\omega_{21}}=\left[\begin{array}{cc} 0& 1\\ 0 & 0 \end{array}\right], \;
y_{\omega_{22}}=\left[\begin{array}{cc} 0& 0 \\ 0 & 1 \end{array}\right].
\]
The matrix $Y_\Omega\in\M_2(\M)$ is not positive, as 
$Y_\Omega=[1]\oplus\left[\begin{array}{cc} 0& 1 \\ 1 & 0 \end{array}\right]\oplus[1]$
has one negative eigenvalue. However, the linear map 
$\mathfrak L_\Omega:\M \rightarrow\M_2(\mathbb C)$ satisfies $\mathfrak L_\Omega(x)=x$,
for every $x\in\M$, and so $\mathfrak L_\Omega$ is a completely positive map; by criterion (\ref{lemma 4.7}), we
deduce that $\Omega\in\M_2(\M_*)_+$. Thus, in identifying matrices 
$[\omega_{ij}]_{i,j}$ over $\M_*$ with the matrices
$[y_{\omega_{ij}}]_{i,j}$ over $\M$, we deduce from this example  that
$\M_2(\M_*)_+\cap \M_2(\M)\not\subseteq \M_2(\M)_+$. 
By similar reasoning, the matrix $\Delta=\left[\begin{array}{cc} \omega_{11}& \omega_{21} \\ \omega_{12}& \omega_{22} \end{array}\right]\in \M_2(\M_*)$
is induced by the positive operator matrix $Y_\Delta\in\M_2(\M)_+$, defined by
\[
Y_\Delta=\left[\begin{array}{cc} y_{\omega_{11}}& y_{\omega_{21}} \\ y_{\omega_{12}}& y_{\omega_{22}} \end{array}\right]
=\left[\begin{array}{cccc} 1&0&0&1 \\ 0&0&0&0 \\ 0&0&0&0 \\ 1&0&0&1 \end{array}\right],
\]
but $ \Delta$ is not positive in  $\M_2(\M_*)$ because $\mathfrak L_\Delta:\M \rightarrow\M_2(\mathbb C)$ satisfies $\mathfrak L_\Delta(x)=x^t$,
for every $x\in \M $, and it is well known that the transpose map fails to be completely positive.

Notwithstanding the discussion of the previous paragraph, it is nevertheless true
that a complete positivity linear map $\mathcal E$ on the matrix-ordered space
$\M_d(\mathbb C)_*=\T_d(\mathbb C)$ (the $d\times d$ complex matrices in the trace norm) 
is also completely positive on the C$^*$-algebra $\M_d(\mathbb C)$. This is a remarkably fortunate
circumstance, as the matrix orders
on $\T_d(\mathbb C)$ and $\M_d(\mathbb C)$ are distinct (as indicated in the example of the previous paragraph), and 
the literature (including the literature on fidelity) makes extensive use of this 
fortunate fact by frequently make little or no reference to
the matrix order on $\T_d(\mathbb C)$ and, instead, drawing entirely upon
the matrix order on $\M_d(\mathbb C)$. 
Although this fortunate circumstance may be relevant for matrix algebras, one does not expect the
same situation to persist with arbitrary von Neumann algebras. Therefore, in what follows, it will be important for us
to distinguish
between the matricial order on the predual $\M_*$ and the matricial order
on $\M$ in discussing completely positive linear maps
on $\M_*$.

With this preface, the following definition is natural.

\begin{definition}
If $\M$ is a von Neumann algebra, then a
\emph{channel}, or \emph{quantum channel}, 
is a continuous linear operator $\mathcal E:\M_*\rightarrow\M_*$
such that the dual map $\mathcal E^*:\M\rightarrow\M$ is
unital, normal, and completely positive.
\end{definition}

Not every completely positive linear map on a von Neumann algebra admits a Kraus decompostion, but those that do
are called inner maps \cite{claire--havet1990}.

\begin{definition} Let $\M$ be a von Neumann algebra with predual $\M_*$.
\begin{enumerate}
\item A completely positive linear map $\phi:\M\rightarrow\M$ is \emph{inner} if there exists a sequence 
$\{a_k\}_{k\in\mathbb N}$ in $\M$ such that $\phi(x)=\displaystyle\sum_{k=1}^\infty a_k^*xa_k$, for every $x\in\M$, 
where the convergence of the sum is with respect to the ultraweak topology of $\M$; that is, for each $x\in \M$,
\[
\omega\left(\phi(x)\right)=\lim_{m\rightarrow\infty}\sum_{k=1}^m\omega( a_k^*xa_k),
\]
for every $\omega\in\M_*$.
\item A channel $\mathcal E:\M_*\rightarrow\M_*$ is \emph{inner} if the unital completely positive linear map
$\mathcal E^*$ is inner.
\end{enumerate}
\end{definition}

%%%%%%%%%%%%
\subsection{Fidelity}

Semifinite von Neumann algebras (see \cite{Takesaki-bookI,Takesaki-bookII} )
admit a very natural analogue of the classical notion of density operator.
Therefore, assume for the remainder of this paper
that $\M$ is a semifinite  von Neumann algebra, and that 
$\tau$ is a fixed faithful normal semifinite tracial weight on $\M$. 
Because we have already encountered finite von Neumann
algebras in the previous section, we shall also assume that $\M$ is not finite. 
Thus, the trace of the identity $1\in\M$ is
infinite.

By definition, $\tau$ is a function
$\M_+\rightarrow[0,+\infty]$ such that, for all $a,b\in\M_+$, $x\in \M$, and $\lambda\geq0$ in $\mathbb R$, we have
$\tau(x^*x)=\tau(xx^*)$,
$\tau(a+b)=\tau(a)+\tau(b)$, $\tau(\lambda a)=\lambda\tau(a)$, $\tau(a)>0$ 
if $a\not=0$, $\tau\left(\sup_\alpha a_\alpha\right)=\sup\tau(a_\alpha)$
for every bounded  increasing net $\{a_\alpha\}_\alpha$ in $\M_+$, and for each nonzero $h\in\M_+$ there exists nonzero $h_0\in\M_+$
with $h_0\leq h$ and $\tau(h_0)<\infty$.
As shown in \cite{fack1982}, if $z\in\M$, then 
\begin{equation}\label{E:trace}
\tr(|z|)=\int_0^{\infty}\mu_z(t)\,dt,
\end{equation}
where
for, each $t\in[0,\infty)$,
\[
\mu_z(t)=\mbox{inf}\,\left\{\|ze\|\,|\,e\in\mathcal P(\M),\,\tr(1-e)\le t
\right\},
\]
and where $\mathcal P(\M)$ is the projection lattice for $\M$. 
Moreover,
$\mu_z=\mu_{z^*}=\mu_{|z|}$ and, consequently, for any $w,z\in \M$,
\begin{equation}\label{E:symmetry identity}
\mu_{|wz^*|}=\mu_{|zw^*|}.
\end{equation}
Furthermore, if $h\in \M_+$ and if 
$\psi:[0,\infty)\rightarrow[0,\infty)$ is an increasing continuous function
such that $\psi(0)=0$, then 
\begin{equation}\label{E:functional calculus}
\mu_{\psi(h)}(t)=\psi\left(\mu_h(t)\right),\mbox{ for all }t\in[0,\infty) .
\end{equation}
Using $\psi(t)=\sqrt{t}$, 
equations \eqref{E:trace}, \eqref{E:symmetry identity}, \eqref{E:functional calculus}
imply that 
\begin{equation}\label{E:symmetry}
\tau(|h^{1/2}k^{1/2}|)=\tau(|k^{1/2}h^{1/2}|) 
\end{equation}
for all $h,k\in \M_+$. %Hence, in particular, with $h=\pi_\tau(\sigma)$ and $k=\pi_\tau(\rho)$, we obtain
%$F_\tau(\sigma,\rho)=F_\tau(\rho,\sigma)$, which proves (1).

Define: 
\[
\mathfrak n_\tau = \{x\in\M\,|\,\tau(x^*x)<\infty\}
\,\mbox{ and }  \,
\mathfrak m_\tau= (\mathfrak n_\tau)^2 = \left\{\sum_{j=1}^k  x_jy_j\,|\, k\in\mathbb N,\,x_j,y_j\in\mathfrak n_\tau \right\}.
\]
The sets $\mathfrak m_\tau$ and $\mathfrak m_\tau$ are (algebraic) ideals of $\M$. If $\mathfrak p_\tau\subseteq\M_+$ is
the set of all $a\in\M_+$ such that $\tau(a)<\infty$, then $\mathfrak p_\tau=\mathfrak m_\tau\cap \M_+$ and the
function $\tau_{\vert \mathfrak p_\tau}$ extends to a linear map, which we denote again by $\tau$,  on $\mathfrak m_\tau$
such that $\tau(x^*)=\overline{\tau(x)}$, $\tau(xy)=\tau(yx)$ for all $x\in\M$ and $y\in\mathfrak m_\tau$, and $\tau(xy)=\tau(yx)$
for all $x,y\in\mathfrak n_\tau$. The ideal $\mathfrak m_\tau$ is
called the \emph{trace ideal} of $\tau$. If $\M=\B(\H)$ and if $\tau$ is the canonical trace on $\B(\H)$, then 
$\mathfrak m_\tau=\T(\H)$. Hence, the elements of the ideal $\mathfrak m_\tau$ are analogues of
trace-class operators.

Recall from \cite{nelson1974} and \cite[Chapter IX]{Takesaki-bookII} 
that $\H$ and $\M$ determine a topological vector space $\mathfrak M(\H)$
and a topological involutive algebra $\mathfrak M(\M)$ such that $\H$ and $\M$ are 
dense (in appropriate topologies) in $\mathfrak M(\H)$ and $\mathfrak M(\M)$ respectively, 
and that $\mathfrak M(\M)$ acts on $\mathfrak M(\H)$ in a natural fashion.
In this context, for each $z\in\mathfrak M(\M)$ and $\varepsilon>0$ there exists a projection $p\in\M$ such that 
$zp\in\M$ and $\tau(1-p)<\varepsilon$. 
The set $\mathfrak M(\M)_+$ of all $z^*z$, for $z\in\mathfrak M(\M)$, is a pointed convex cone
and for each $a\in\mathfrak M(\M)_+$ there is a unique $b\in\mathfrak M(\M)_+$ (denoted by $a^{1/2}$) for which $b^2=a$. 
Thus, the element $|z|$ given by $(z^*z)^{1/2}$ lies in $\mathfrak M(\M)$ for each $z\in\mathfrak M(\M)$.
Furthermore, the function $\tau$
extends to $\mathfrak M(\M)_+$ via $\tau(a)=\displaystyle\lim_{\varepsilon\rightarrow 0^+}\tau\left( a(1+\varepsilon a)^{-1}\right)$, for
$a\in\mathfrak M(\M)_+$, and satisfies the usual trace properties: $\tau(z^*z)=\tau(zz^*)$,
$\tau(a+b)=\tau(a)+\tau(b)$, and $\tau(\lambda a)=\lambda\tau(a)$ for all $z\in \mathfrak M(\M)$, $a,b\in \mathfrak M(\M)_+$, and $\lambda\in\mathbb R$
with $\lambda\ge0$. 

The predual $\M_*$  is linearly order isomorphic to
set $\left\{z\in\mathfrak M(\M)\,|\,\tau(|z|)<\infty\right\}$, and so we identify these two sets. 
The map $z\mapsto\tau(|z|)$ defines a norm $\|\cdot\|_1$ on $\M_*$ and with respect
to this norm $\M_*$ is a Banach space containing $\mathfrak m_\tau$ as a norm-dense linear submanifold.
In the case of a type I von Neumann algebra $\M$, 
nothing new is obtained, since $\mathfrak M(\M)=\M$ in this case. In particular, if $\M=\B(\H)$, 
then $\mathfrak M(\M)=\M$ and $\M_*=\mathfrak m_\tau=\T(\H)$.

For each $n\in\mathbb N$, let $\tau^{(n)}:\M_n(\mathfrak m_\tau)\rightarrow\M_n(\mathbb C)$ denote the
linear map
\[
\tau^{(n)}(Y)=\left[\tau(y_{ij}) \right]_{i,j},
\]
for $Y=[y_{ij}]_{i,j=1}^n\in\M_n(\mathfrak m_\tau)$.
 
\begin{lemma}\label{basic fact 3} If $Y\in\M_n(\mathfrak m_\tau)$ is a positive operator
on $\displaystyle\bigoplus_1^n \H$, then $\tau^{(n)}(Y)$ is a positive operator on $\mathbb C^n$.
\end{lemma}

\begin{proof} Recall that $\mathfrak p_\tau=\mathfrak m_\tau\cap\M_+$ \cite[Lemma V.2.16]{Takesaki-bookI} and that
if $\omega$ is a normal state on $\M$, then $\omega$ is a
completely positive linear map of $\M$. Because the trace $\tau:\M_+\rightarrow[0,\infty]$ is normal, it is a sum of a family of 
normal states \cite[p.~332]{Takesaki-bookI}; hence, $\tau^{(n)}(X)$ is a (possibly nonconvergent) 
sum of positive elements of $\M_n(\mathbb C)_+$. Therefore, 
because $Y\in\M_n(\mathfrak m_\tau)\cap \M_n(\M)_+$, the complex matrix $\tau^{(n)}(Y)$ is a 
(convergent) sum
of positive elements and is, hence, positive.
\end{proof}

\begin{lemma}\label{in ideal} If $a,b\in\mathfrak m_\tau\cap \M_+$, then $|a^{1/2}b^{1/2}|\in\mathfrak m_\tau$.
\end{lemma}

\begin{proof} Recall from equation (\ref{E:trace}) that
$\displaystyle\tau(|z|)=\int_0^{\infty}\mu_z(t)\,dt$,
for every $z\in M$.Since $a$ and $b$
are positive, we have, for every $s>0$ in $\mathbb R$, that
\[
\begin{array}{rcl}
\displaystyle\int_0^s \mu_{a^{1/2}b^{1/2}}(t)\,dt &\leq& \displaystyle\int_0^s \mu_{a^{1/2}}(t)\mu_{b^{1/2}}(t)\,dt \\ && \\
&=&\displaystyle\int_0^s \mu_{a}(t)^{1/2}\mu_{b} (t)^{1/2}\,dt \\ && \\
&\leq& \left(\displaystyle\int_0^s \mu_{a}(t)\,dt\right)^{1/2}  \left(\displaystyle\int_0^s \mu_{b}(t)\,dt \right)^{1/2} \\ && \\
&\leq& \sqrt{\tau(a)\tau(b)},
\end{array}
\]
where the first of the inequalities above is a consequence of \cite[Corollaire 4.4]{fack1982}.
Thus, $\tau(|a^{1/2}b^{1/2}|)\leq \sqrt{\tau(a)\tau(b)}$, which implies that $|a^{1/2}b^{1/2}|\in\mathfrak m_\tau$.
\end{proof}

\begin{definition} A \emph{density operator in $\M$} 
is a positive operator $\rho\in\M$ such that $\tau(\rho)=1$.
\end{definition} 

Let $\mathfrak s_\tau$ denote the set of all density operators in $\M$.
If $\mathcal E:\M_*\rightarrow\M_*$ is a channel, then $\mathcal E$ is
trace preserving and so $\mathcal E$ maps $\mathfrak s_\tau$ back into itself. Because $\mathfrak s_\tau$ spans
$\mathfrak m_\tau$, we deduce that
\[
\mathcal E(\mathfrak m_\tau)\subseteq\mathfrak m_\tau,
\]
for every channel $\mathcal E:\M_*\rightarrow\M_*$.

\begin{lemma}\label{inner cp} If  a channel $\mathcal E:\M_*\rightarrow\M_*$ is inner, then 
\begin{enumerate} 
\item there exists a sequence 
$\{a_k\}_{k\in\mathbb N}$ in $\M$ such that $\displaystyle\sum_{k=1}^\infty a_k^*   a_k=1$, and
$\mathcal E(\rho)=\displaystyle\sum_{k=1}^\infty a_k \rho a_k^*$, for every density operator $\rho\in\M$, 
\item $\mathcal E^{(2)}\left(\M_2(\mathfrak m_\tau)\cap \M_2(\M)_+\right)\subseteq \M_2(\M)_+$, and
\item $\mathcal E(y)^*\mathcal E(y)\leq\mathcal E(y^*y)$, for every $y\in\mathfrak m_\tau$.
\end{enumerate}
\end{lemma}

\begin{proof} By definition of inner, 
there exists a sequence 
$\{a_k\}_{k\in\mathbb N}$ in $\M$ such that $\displaystyle\sum_{k=1}^\infty a_k^*   a_k=1$, and
$\mathcal E^*(x)=\displaystyle\sum_{k=1}^\infty a_k^*x a_k$, for every $x\in\M$, 
where the convergence of the sum is with respect to the ultraweak topology of $\M$. Therefore, if $\rho\in\mathfrak s_\tau$
and $x\in \M$, then
\[
\tau\left(\mathcal E(\rho)x\right)
=\tau\left(\rho\mathcal E^*(x)\right)
=\lim_{m\rightarrow\infty}\sum_{k=1}^m\tau( \rho a_k^*xa_k) 
=\lim_{m\rightarrow\infty}\sum_{k=1}^m\tau(  xa_k\rho a_k^*) 
=\lim_{m\rightarrow\infty}\tau\left(x\sum_{k=1}^m a_k\rho a_k^*\right),
\]
and so 
$\mathcal E(\rho)=\displaystyle\sum_{k=1}^\infty a_k \rho a_k^*$, which proves the first statement.
The second and third statements follow immediately from the inner structure of $\mathcal E$.
\end{proof}

By Lemma \ref{in ideal}, if $\sigma,\rho\in\mathfrak s_\tau$, then $|\sigma^{1/2}\rho^{1/2}|$ has finite trace; thus,
we may define fidelity for pairs of density operators in $\M$.
 
\begin{definition} The \emph{fidelity} of a pair of density operators
$\sigma,\rho\in \M$ is the
nonnegative real number $F_\tau(\sigma,\rho)$ defined by
\[
F_\tau(\sigma,\rho)=\tau\left(|\sigma^{1/2}\rho^{1/2}|\right).
\]
\end{definition}

The following is the main result on the fidelity of density operators.

\begin{theorem}[Fidelity in Semifinite von Neumann Algebras]\label{fidelity vN alg}  \hfill
\begin{enumerate}
\item (Basic Properties of Fidelity)
If $\sigma,\rho\in \mathfrak s_\tau$, then
\begin{enumerate}
\item $F_\tau(\sigma,\rho)=F_\tau(\rho,\sigma)$, 
\item $0\leq F_\tau(\sigma,\rho) \leq 1$,
\item $F_\tau(\sigma,\rho)=0$ if and only if $\sigma \bot \rho$, and
\item $F_\tau(\sigma,\rho)=1$ if and only if $\sigma=\rho$.
\end{enumerate} 
\item (Monotonicity of Fidelity) If $\mathcal E:\M_*\rightarrow\M_*$ is a linear map
such that $\mathcal E^{(2)}\left(\M_2(\mathfrak m_\tau)\cap \M_2(\M)_+\right)\subseteq \M_2(\M)_+$, then
$\mathcal E$ is a positive and
\begin{equation}\label{ie:mf sf}
F_\tau(\sigma,\rho)\leq F_\tau\left({\mathcal E}(\sigma),{\mathcal E}(\rho)\right),
\end{equation}
for all $\sigma,\rho\in \mathfrak s_\tau$.
\item (Preservation of Fidelity) If $\M$ is a factor and if $\mathcal E:\M_*\rightarrow\M_*$ is a  
bijective positive linear map such that  
\begin{equation}\label{e:fp sf}
F_\tau(\sigma,\rho)= F_\tau\left({\mathcal E}(\sigma),{\mathcal E}(\rho)\right),
\end{equation}
for all $\sigma,\rho\in \mathfrak s_\tau$, then $ {\mathcal E}^*$ is an automorphism of $\M$.
\end{enumerate}
\end{theorem}

\begin{proof}  Equation (\ref{E:symmetry}) shows that $F_\tau(\sigma,\rho)=F_\tau(\rho,\sigma)$, whereas
\cite[Corollary 3.2]{farenick--manjegani2005} yields $0\leq F_\tau(\sigma,\rho) \leq 1$. It is clear that
$F_\tau(\sigma,\rho)=0$ if and only if $\sigma \bot \rho$, while 
\cite[Theorem 3.4]{farenick--manjegani2005} shows that $F_\tau(\sigma,\rho)=1$ if and only if $\sigma=\rho$.

To establish monotonicity of fidelity, we again draw upon a method of proof given in Watrous's monograph \cite{Watrous-book}. First, however, note that the hypothesis 
$\mathcal E^{(2)}\left(\M_2(\mathfrak m_\tau)\cap \M_2(\M)_+\right)\subseteq \M_2(\M)_+$
immediately yields the positivity of $\mathcal E$ by considering positive $2\times 2$ matrices of the form
$\left[\begin{array}{cc} y & 0 \\ 0 & 0\end{array}\right]$, for $y\in\mathfrak m_\tau\cap \M_+$.

Suppose that $\sigma,\rho\in\mathfrak s_\tau$ and $x\in \M$, and consider the matrix 
$X=\left[\begin{array}{cc} \sigma&x\\ x^*&\rho\end{array}\right]$ in $\M_2(\M)$. By Lemma \ref{matrix lemma},
the matrix $X\in \M_2(\M)_+$ if and only if $x=\sigma^{1/2}y\rho^{1/2}$ for some $y\in\M$ with $\|y\|\leq 1$.
Hence,  
\[
\begin{array}{rcl}
\sup\left\{|\tau(x)|\,|\, \left[\begin{array}{cc} \sigma&x\\ x^*&\rho\end{array}\right]\in M_2(\mathfrak m_\tau)_+\right\} &=& 
\sup\left\{|\tau(\rho^{1/2}\sigma^{1/2}y)\,|\,y\in\M,\;\|y\|\leq 1\right\} \\ && \\
&=& \|\rho^{1/2}\sigma^{1/2}\|_1  \\ && \\
&=&\tau\left( |\rho^{1/2}\sigma^{1/2}|\right)  \\ && \\
&=& F_\tau(\sigma,\rho).
\end{array}
\]

We shall show that a similar formula applies to the fidelity of the density operator
$\mathcal E(\sigma)$ and $\mathcal E(\rho)$. To this end,
fix $y\in\M$ with $\|y\|\leq 1$ and let $x=\sigma^{1/2}y\rho^{1/2}$. 
Thus, the matrix $\left[\begin{array}{cc} \sigma&x\\ x^*&\rho\end{array}\right]$
in $\M_2(\mathfrak m_\tau)$ belongs to the cone $\M_2(\M)_+$. Therefore, by the hypothesis on ${\mathcal E}$, 
we deduce that
\[
\left[\begin{array}{cc} {\mathcal E}(\sigma)&{\mathcal E}(x)\\ {\mathcal E}(x^*)&{\mathcal E}(\rho)\end{array}\right]  
\in  \M_2(\M)_+.
\]
Hence, Lemma \ref{basic fact 3} yields 
\[
\left[\begin{array}{cc} \tau\left({\mathcal E}(\sigma)\right)&\tau\left({\mathcal E}(x)\right)\\ \tau\left({\mathcal E}(x^*)\right)&\tau\left({\mathcal E}(\rho)\right)\end{array}\right]  =
\tau^{(2)}\left(\left[\begin{array}{cc} {\mathcal E}(\sigma)&{\mathcal E}(x)\\ {\mathcal E}(x^*)&{\mathcal E}(\rho)\end{array}\right]\right)
\in \M_2(\mathbb C)_+,
\]
and so $|\tau(x)|=|\tau({\mathcal E}(x))|\leq F_\tau\left({\mathcal E}(\sigma),{\mathcal E}(\rho)\right)$. Therefore, the
supremum of the real numbers $|\tau(x)|$
over all $x\in\mathfrak m_\tau$ 
for which $\left[\begin{array}{cc} \sigma&x\\ x^*&\rho\end{array}\right]\in \M_2(\M)_+$ is also bounded above by 
$ F_\tau\left({\mathcal E}(\sigma),{\mathcal E}(\rho)\right)$, which implies that
$F_\tau(\sigma,\rho)\leq F_\tau\left({\mathcal E}(\sigma),{\mathcal E}(\rho)\right)$.  

Lastly, assume that $ {\mathcal E}$ is a positive linear bijection that preserves fidelity.
Thus, $\mathcal E^*$ is an invertible operator on $\M$.
 If $x\in \M_+$ and $y=( {\mathcal E}^*)^{-1}(x)$, then for every $\rho\in\mathfrak s_\tau$,
 \[
 \tau\left( y\rho\right) = \tau\left( ( {\mathcal E}^*)^{-1}(x) \rho\right)=\tau\left( x {\mathcal E}^{-1}(\rho)\right)\geq0.
 \]
 That is, $\omega(y)\geq0$ for every normal state $\omega$ on $\M$; hence,
$y\in\M_+$, which proves that $( {\mathcal E}^*)^{-1}$ is positive linear map. Thus, $ {\mathcal E}^*$ is a linear order isomorphism.
 By \cite[Theorem 5]{kadison1951}, $ {\mathcal E}^*$ is a Jordan isomorphism. 
 Because $\mathcal E^*$ satisfies the Schwarz inequality $\mathcal E^*(x^*x)\geq\mathcal E^*(x^*)\mathcal E^*(x)$ for all $x\in\M$, the Jordan isomorphism $\mathcal E^*$ is in fact an
 automorphism \cite[Corollary 3.6]{stormer1965}.
 \end{proof}

 \begin{corollary}\label{i} If ${\mathcal E}$ is an inner channel, then
 $F_\tau(\sigma,\rho)\leq F_\tau\left({\mathcal E}(\sigma),{\mathcal E}(\rho)\right)$,
for all $\sigma,\rho\in \mathfrak s_\tau$. 
 \end{corollary}

\begin{proof} By Lemma \ref{inner cp}, $\mathcal E$
satisfies $\mathcal E^{(2)}\left(\M_2(\mathfrak m_\tau)\cap \M_2(\M)_+\right)\subseteq \M_2(\M)_+$.
Therefore, by  Theorem \ref{fidelity vN alg}, $\mathcal E$ has the monotonicity property for fidelity.
\end{proof}

In the classical setting of type I factors, the following known result \cite{molnar2001} is recovered.

 \begin{corollary}[Moln\'ar]\label{mlnr} 
 If ${\mathcal E}:\T(\H)\rightarrow\T(\H)$ is a  channel, then
 $F_\tau(\sigma,\rho)\leq F_\tau\left({\mathcal E}(\sigma),{\mathcal E}(\rho)\right)$,
for all density operators $\sigma,\rho\in \T(\H)$. Furthermore, if
 ${\mathcal E}:\T(\H)\rightarrow\T(\H)$ is  surjective channel and
 preserves fidelity, then ${\mathcal E}$ is a unitary channel.
 \end{corollary}
 
 \begin{proof} In this case, $\M=\B(\H)$ and
 $\M_*=\mathfrak m_\tau=\T(\H)$. 
 By Kraus's theorem \cite{kraus1971}, every channel on $\T(\H)$ is inner. Hence, 
 $\mathcal E$ has the monotonicity property for fidelity.
 Moreover, 
 the proof of Lemma \ref{order iso} is valid in $\T(\H)$,
 and so $\mathcal E$ is injective. Therefore, coupled with the hypothesis that $\mathcal E$ is surjective, 
 Theorem \ref{fidelity vN alg} indicates that ${\mathcal E}^*$ is an automorphism of $\B(\H)$. Hence, there is a 
 unitary $u\in\B(\H)$ such that ${\mathcal E}^*(x)=u^*xu$ for all $x\in\B(\H)$. 
 That is, ${\mathcal E}(s)=usu^*$ for every $s\in\T(\H)$.
 \end{proof}

%%%%%%%%%%%%%%
\section{Discussion}

The principal achievements of this paper broaden the context in which the classical notion of density operator 
has meaning. In the first of these extensions the ambient operator algebra is a unital C$^*$-algebra equipped with a
faithful tracial functional, whereas in the second type of extension the ambient operator algebra is a semifinite von Neumann algebra.
Examples of such operator algebras are described below.

While the full matrix algebra $\M_d(\mathbb C)$ of $d\times d$ complex matrices is one such example,
there are many other examples that are infinite-dimensional. One of the most important of such examples is the Fermion algebra
$\A=\displaystyle\bigotimes_1^\infty \M_2(\mathbb C)$, which is the C$^*$-algebra representing the infinite canonical anticommutation relations
arising from the quantum mechanical study of fermions and which has a (unique) faithful trace functional $\tau$ that arises from the normalised 
canonical trace on $\M_{2k}(\mathbb C)\cong\displaystyle\bigotimes_1^k\M_2(\mathbb C)$, for all $k\in\mathbb N$ \cite[\S11.9]{Farenick-book2}.
Thus, every positive $\tau$-preserving linear map $\mathcal E:\A\rightarrow\A$ satisfies the monotonicity property of fidelity, by
Theorem \ref{miza}. Furthermore, because the Fermion C$^*$-algebra $\A$  is also finite and quasi-transitive \cite[Theorem 5.10]{rickart1946},
a surjective fidelity-preserving Schwarz map $\mathcal E:\A\rightarrow\A$ is necessarily an 
automorphism of $\A$, by Theorem \ref{main result 1}.

In considering the setting of von Neumann algebras, one of course has the classical case of $\B(\H)$ and its predual $\T(\H)$. In this setting, a (classical)
density operator is a positive (compact) operator $\rho$ on $\H$ for which the sequence of eigenvalues of 
$\rho$ (repeated according to geometric multiplicity) is summable. These type I factors are precisely the von Neumann algebras that arise
in ordinary quantum mechanics.

Quantum statistical mechanics involving (infinitely) many particles is modelled mathematically on von Neumann algebras of type II and III. That is,
when the number of particles is very large, then a mathematical model that treats large as infinite is adopted, leading to so-called infinite systems.
A nice explanation of this approach is provided by Kadison in \cite[p.~13394]{kadison1998}. 
Of particular interest are those von Neumann algebras that arise through a limiting process involving finite quantum systems.
A specific example is afforded by
the hyperfinite II${}_1$-factor $\osr$, which is constructed from the Fermion algebra as follows. Consider the 
trace $\tau$ on the Fermion algebra $\A=\displaystyle\bigotimes_1^\infty \M_2(\mathbb C)$, normalised so that $\tau(1)=1$. 
By the Gelfand-Naimark-Segal theorem,
there exist a unital homomorphism $\pi:\A\rightarrow\B(\H)$ and a unit vector $\xi\in \H$ such that $\{\pi(a)\xi\,:\,a\in\A\}$ is dense in $\H$ and
$\tau(a)=\langle\pi(a)\xi,\xi\rangle$, for every $a\in\A$.
If $\osr$ denotes the closure of $\pi(\A)$ in $\B(\H)$ with respect to the strong operator topology, then the trace $\tau$ extends to a faithful trace on $\osr$
(denoted by $\tau$ again) such that $\tau(x)=\langle x\xi,\xi\rangle$, for every $x\in \osr$ \cite[Corollary 11.105]{Farenick-book2}.
The von Neumann algebra $\osr$ is a finite factor (and, hence, quasi-transitive), which means that the C$^*$-algebra results obtained herein apply to $\osr$.
In particular, every surjective channel on $\osr$ that preserves fidelity is necessarily an automorphism of $\osr$. 
 
With the semifinite but non-finite setting, continue with the notation of the previous paragraph and let $\M$ denote the closure, 
with respect to the strong operator topology, of the linear span $\M_0$ all operators acting on the Hilbert space $\H\otimes\H$
of the form $\rho\otimes y$, where $\rho\in\T(\H)$ and $y\in\osr$ are arbitrary. The linear map $\M_0\rightarrow\mathbb C$
that sends each $\rho\otimes y$ to $\tr(\rho)\tau(y)$ extends to a tracial weight $\mbox{\rm Tr}_{\M}$ on $\M$. In particular, linear combinations
of finitely many elements of the form $\rho\otimes y$, where $\rho\in\T(\H)$ and $y\in\osr$, lie in the trace ideal $\mathfrak m_{\mbox{Tr}_{\M}}$.

Finally, in a different direction, the effort undertaken herein to establish the monotonicity of
fidelity is substantially greater than the effort required to prove some of the other basic properties of fidelity. 
This fact of effort is very reminiscent of recent work
of M\"uller-Hermes and Reeb \cite{muller-hermes--reeb2016} in proving the monotonicity of quantum relative entropy for trace-preserving
positive linear maps on $\T(\H)$. Although Schwarz maps are positive but not nessecarily $2$-positive, it would be of interest to know whether
our main results on fidelity preservation, which rely upon sufficient level of
positivity as embodied by the Schwarz inequality, hold for arbitrary trace-preserving positive linear maps.

%%%%%%%%%%%%%%
\section*{Acknowledgements}

This work of the first author is supported in part by the Discovery Grant program of
the Natural Sciences and Engineering Research Council of Canada. 
We thank Sarah Plosker for introducing us to the use of fidelity in quantum information theory and 
for bringing Moln\'ar's work  \cite{molnar2001} to our attention.

%%%%%%%%%%%%%%%%%%%%%%%%%%% bibliography %%%%%%%%%%%%%%%%%%%%%%%%%
%\bibliographystyle{abbrv}
%%%%\bibliographystyle{amsplain}
%\bibliography{doug-refs}
%\end{document}

%%%%%%%%%%%%%%%%%%%%%%%%%%% bibliography %%%%%%%%%%%%%%%%%%%%%%%%%

\end{document}